\documentclass[a4paper]{article}
\usepackage[utf8]{inputenc}

\usepackage[margin=1.3in]{geometry}

\usepackage{amsthm,amssymb,amsfonts,amsmath}
\usepackage{amssymb,amsfonts,amsmath}
\usepackage{xcolor}
\usepackage{mathabx}
\newtheorem{theorem}{Theorem}
\newtheorem{lemma}[theorem]{Lemma}
\newtheorem{corollary}[theorem]{Corollary}

\newtheorem{conjecture}{Conjecture}

\usepackage{graphicx}
\graphicspath{{figures/}}

\usepackage{thm-restate}
\usepackage{url}
\usepackage{faktor}
\usepackage{tikz-cd}

\title{Equipartitions with Wedges and Cones}

\author{Patrick Schnider\thanks{Department of Mathematical Sciences,
        University of Copenhagen, Denmark. {\tt ps@math.ku.dk}. Has received funding from the European Research Council under the European Unions Seventh Framework Programme ERC Grant agreement ERC StG 716424 - CASe. Part of this work was done when the author was employed at ETH Z\"{u}rich.}}
        
\date{}

\begin{document}

\maketitle

\begin{abstract}
A famous result about mass partitions is the so called \emph{Ham-Sandwich theorem}.
It states that any $d$ mass distributions in $\mathbb{R}^d$ can be simultaneously bisected by a single hyperplane.
In this work, we study two related questions.

The first one is how many masses we can simultaneously partition with a $k$-fan, that is, $k$ half-hyperplanes in $\mathbb{R}^d$, emanating from a common $(d-2)$-dimensional apex.
This question was extensively studied in the plane, but in higher dimensions the only known results are for the case where $k$ is an odd prime.
We extend these results to a larger family of values of $k$.
We further present a new result for $k=2$, which generalizes to cones.

The second question considers bisections with double wedges or, equivalently, Ham-Sandwich cuts after projective transformations.
Here we prove that given $d$ families of $d+1$ point sets each, there is always a projective transformation such that after the transformation, each family has a Ham-Sandwich cut.
We further prove a result on partitions with parallel hyperplanes after a projective transformation.

%

All of our results are proved using topological methods.
We use some well-established techniques, but also some newer methods.
In particular, we introduce a Borsuk-Ulam theorem for flag manifolds, which we believe to be of independent interest.
\end{abstract}


\section{Introduction}
Equipartitions of point sets and mass distributions are essential problems in combinatorial geometry.
The general goal is the following: given a number of masses in some space, we want to find a dissection of this space into regions such that each mass is evenly distributed over all the regions, that is, each region contains the same amount of this mass.
Usually, the underlying space considered is a sphere or Euclidean space, and the regions are required to satisfy certain conditions, for example convexity.
Arguably the most fundamental result about equipartitions is the \emph{Ham-Sandwich theorem} (see e.g.\ \cite{Matousek, StoneTukey}, Chapter 21 in \cite{Handbook}), which states that any $d$ mass distributions in $\mathbb{R}^d$ can be simultaneously bisected by a single hyperplane.
A \emph{mass distribution} $\mu$ in $\mathbb{R}^d$ is a measure on $\mathbb{R}^d$ such that all open subsets of $\mathbb{R}^d$ are measurable, $0<\mu(\mathbb{R}^d)<\infty$ and $\mu(S)=0$ for every lower-dimensional subset $S$ of $\mathbb{R}^d$.
The result is tight in the sense that there are collections of $d+1$ masses in $\mathbb{R}^d$ that cannot be simultaneously bisected by a single hyperplane.

However, the restriction that the regions of the dissection should be half-spaces is quite strong, maybe we can bisect more masses by relaxing the conditions on the regions?
Also, the Ham-Sandwich theorem only provides a dissection of the space into two regions, and it is a natural question whether similar statements can be found for more than two regions.
Many variants of mass partitions have been studied, see \cite{SoberonSurvey} for a recent survey.

This paper is mainly motivated by partitions with two objects: fans and double wedges.
In the following, we briefly give an overview over these objects and our results.

\subsection{Fans}

A \emph{$k$-fan} in $\mathbb{R}^d$ is defined by a $(d-2)$ dimensional flat $a$, which we call \emph{apex}, and $k$ semi-hyperplanes emanating from it.
If the apex contains the origin, we call the fan a \emph{$k$-fan through the origin}.
Each $k$-fan partitions $\mathbb{R}^d$ into $k$ wedges, which can be given a cyclic order $W_1,\ldots, W_k$.
For $\alpha=(\alpha_1,\ldots,\alpha_k)$ with $\alpha_1+\ldots+\alpha_k=1$ we say that a $k$-fan \emph{simultaneously $\alpha$-partitions} the mass distributions $\mu_1,\ldots,\mu_{n}$ if $\mu_i(W_j)=\alpha_j\cdot\mu_i(\mathbb{R}^d)$ for every $i\in\{1,\ldots,n\}$ and $j\in\{1,\ldots,k\}$.
See Figure \ref{Fig:fans} for an illustration.

\begin{figure}
\centering
\includegraphics[scale=0.7]{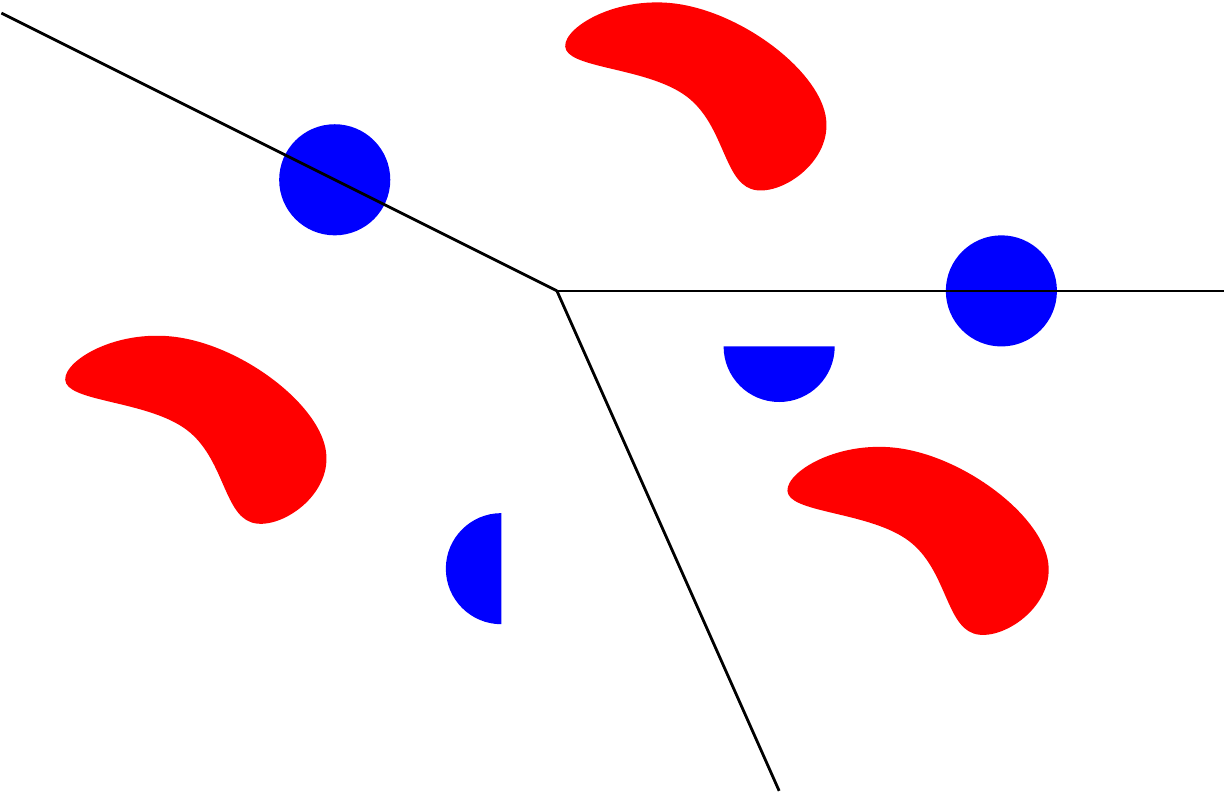}
\caption{Two masses simultaneously $(\frac{1}{3},\frac{1}{3},\frac{1}{3})$-partitioned by a 3-fan.}
\label{Fig:fans}
\end{figure}

Equipartitions with $k$-fans in the plane have been extensively studied, see \cite{Barany4, Barany3, Barany2, Barany, Bereg2, Blagojevic2, Blagojevic, ZivFans}.
In higher dimensions, to the author's knowledge, the only result is due to Makeev, who has shown that if $p$ is an odd prime and $2d-1\geq m(p-1)$, then any $m+1$ mass distributions $\mathbb{R}^d$ can be simultaneously $(\frac{1}{p},\ldots,\frac{1}{p})$-partitioned by a $p$-fan (see Theorem 57 in \cite{KarasevSurvey}).
As there does not seem to be an English proof of this result, we include one in Section \ref{sec:fans}\footnote{The proof presented there is based on ideas communicated to the author by Roman Karasev.}.
However, there seems to be nothing known about partitions with $k$-fans in higher dimensions when $k$ is a composite number.
We give a first such result by extending Makeevs arguments to products of distinct odd primes:

\begin{restatable}{theorem}{fans}
\label{Thm:fans_main}
Let $k=p_1p_2\cdots p_n$ be a product of pairwise distinct odd primes.
\begin{itemize}
\item Let $2d-3\geq m(k-1)$.
Then any $m+1$ mass distributions $\mathbb{R}^d$ can be simultaneously $(\frac{1}{k},\ldots,\frac{1}{k})$-partitioned by a $k$-fan through the origin.
\item Let $2d-1\geq m(k-1)$.
Then any $m+1$ mass distributions $\mathbb{R}^d$ can be simultaneously $(\frac{1}{k},\ldots,\frac{1}{k})$-partitioned by a $k$-fan.
\end{itemize}
\end{restatable}

Further, we also extend Makeevs proof to $(\frac{a_1}{p},\ldots,\frac{a_q}{p})$-partitions with $q$-fans for natural numbers $a_1,\ldots,a_q$ which sum up to $p$, where $p$ is an odd prime, extending planar results from \cite{Barany2}.

On the other hand, Makeevs proof only works for odd primes, and the statement for $p=2$ is actually false.
Indeed, there are examples of $d+2$ masses that cannot be simultaneously bisected by a $2$-fan (i.e., a \emph{wedge}), as we will see soon.
We will also fill in the gap on $2$-fans by showing that $d+1$ masses can always be simultaneously bisected by a wedge.
In fact, we prove a more general statement about bisections with \emph{$k$-cones}.

Recall the definition of a cone: a (spherical) cone in $\mathbb{R}^d$ is defined by an \emph{apex} $a\in\mathbb{R}^d$, a \emph{central axis} $\overrightarrow{\ell}$, which is a one-dimensional ray emanating from the apex $a$, and and angle $\alpha$.
The cone is now the set of all points that lie on a ray $\overrightarrow{r}$ emanating from $a$ such that the angle between $\overrightarrow{r}$ and $\overrightarrow{\ell}$ is at most $\alpha$.
Note that for $\alpha=90^{\circ}$, the cone is a half-space.
Also, for $\alpha>90^{\circ}$, the cone is not convex, which we explicitly allow.
Finally note that the complement of a cone is also a cone and that either a cone or its complement are convex.
Let now $H_k$ be some $k$-dimensional linear subspace of $\mathbb{R}^d$ and let $\pi:\mathbb{R}^d\rightarrow H_k$ be the natural projection.
A $k$-cone $C$ is now a set $\pi^{-1}(C_k)$, where $C_k$ is a cone in $H_k$.
The apex $a$ of $C$ is the set $\pi^{-1}(a_k)$, where $a_k$ is the apex of $C_k$.
It has dimension $d-k$.
Again, note that the complement of a $k$-cone is again a $k$-cone and that one of the two is convex.
Also, a $2$-cone is either the intersection or the union of two halfspaces, that is, a wedge\footnote{So, wedges are both $2$-cones and $2$-fans.}.
Further, a $d$-cone is just a spherical cone.
Alternatively, we could also define a $k$-cone by a $(d-k)$-dimensional apex and a $(d-k+1)$-dimensional half-hyperplane $h$ emanating from it.
The $k$-cone would then be the union of points on all $(d-k+1)$-dimensional half-hyperplanes emanating from the apex such that their angle with $h$ is at most $\alpha$.
We say that a $k$-cone $C$ \emph{simultaneously bisects} the mass distributions $\mu_1,\ldots,\mu_{n}$ if $\mu_i(C)=\frac{1}{2}\mu_i(\mathbb{R}^d)$ for all $i\in\{1,\ldots,n\}$.
See Figure \ref{Fig:wedge} for an illustration.
In Section \ref{sec:cones} we will prove the following:

\begin{restatable}{theorem}{cones}
\label{Thm:cones_general}
Let $\mu_1,\ldots,\mu_{d+1}$ be $d+1$ mass distributions in $\mathbb{R}^d$ and let $1\leq k\leq d$.
Then there exists a $k$-cone $C$ that simultaneously bisects $\mu_1,\ldots,\mu_{d+1}$.
\end{restatable}

\begin{figure}
\centering
\includegraphics[scale=0.7]{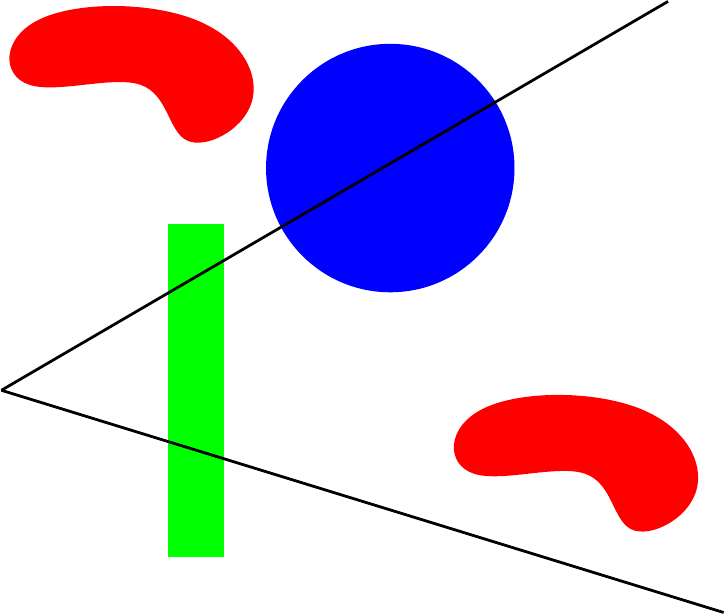}
\caption{Three masses simultaneously bisected by a 2-cone.}
\label{Fig:wedge}
\end{figure}

As mentioned above, this is tight in the sense that there is a family of $d+2$ mass distributions in $\mathbb{R}^d$ that cannot be simultaneously bisected: place $d+1$ point-like masses at the vertices of a $d$-dimensional simplex and a last point-like mass in the interior of the simplex.
Assume that there is a $k$-cone $C$ which simultaneously bisects all the masses and assume without loss of generality that $C$ is convex (otherwise, just consider the complement of $C$).
As it simultaneously bisects all masses, $C$ must now contain all vertices of the simplex, so by convexity $C$ also contains the interior of the simplex, and thus all of the last mass.
Hence $C$ cannot simultaneously bisect all masses.
Still, while we cannot hope to bisect more masses, in some cases we are able to enforce additional restrictions: for $d$-cones in odd dimensions, we can always enforce the apex to lie on a given line (Theorem \ref{Thm:cones_apex_line}).

\subsection{Double wedges}

Another relaxation of Ham-Sandwich cuts was introduced by Bereg et al. \cite{Bereg} and has received some attention lately (see \cite{pizza_cccg, pizza2, pizza1, HSSoCG}): instead of cutting with a single hyperplane, how many masses can we bisect if we cut with several hyperplanes?
In this setting, the masses are distributed into two parts according to a natural 2-coloring of the induced arrangement.
We will only consider bisections with two hyperplanes: let $h_1$ and $h_2$ be two (oriented) hyperplanes. Let $h_1^+$ and $h_1^-$ be the positive and negative side of $h_1$, respectively, and analogous for $h_2$.
We define the \emph{double wedge} $D=(h_1,h_2)$ as the union $(h_1^+\cap h_2^+)\cup (h_1^-\cap h_2^-)$.
Note that the complement of a double wedge is again a double wedge.
We say that a double wedge $D$ simultaneously bisects the mass distributions $\mu_1,\ldots,\mu_{n}$ if $\mu_i(D)=\frac{1}{2}\mu_i(\mathbb{R}^d)$ for all $i\in\{1,\ldots,n\}$.
See Figure \ref{Fig:double_wedge} for an illustration.

\begin{figure}
\centering
\includegraphics[scale=0.7]{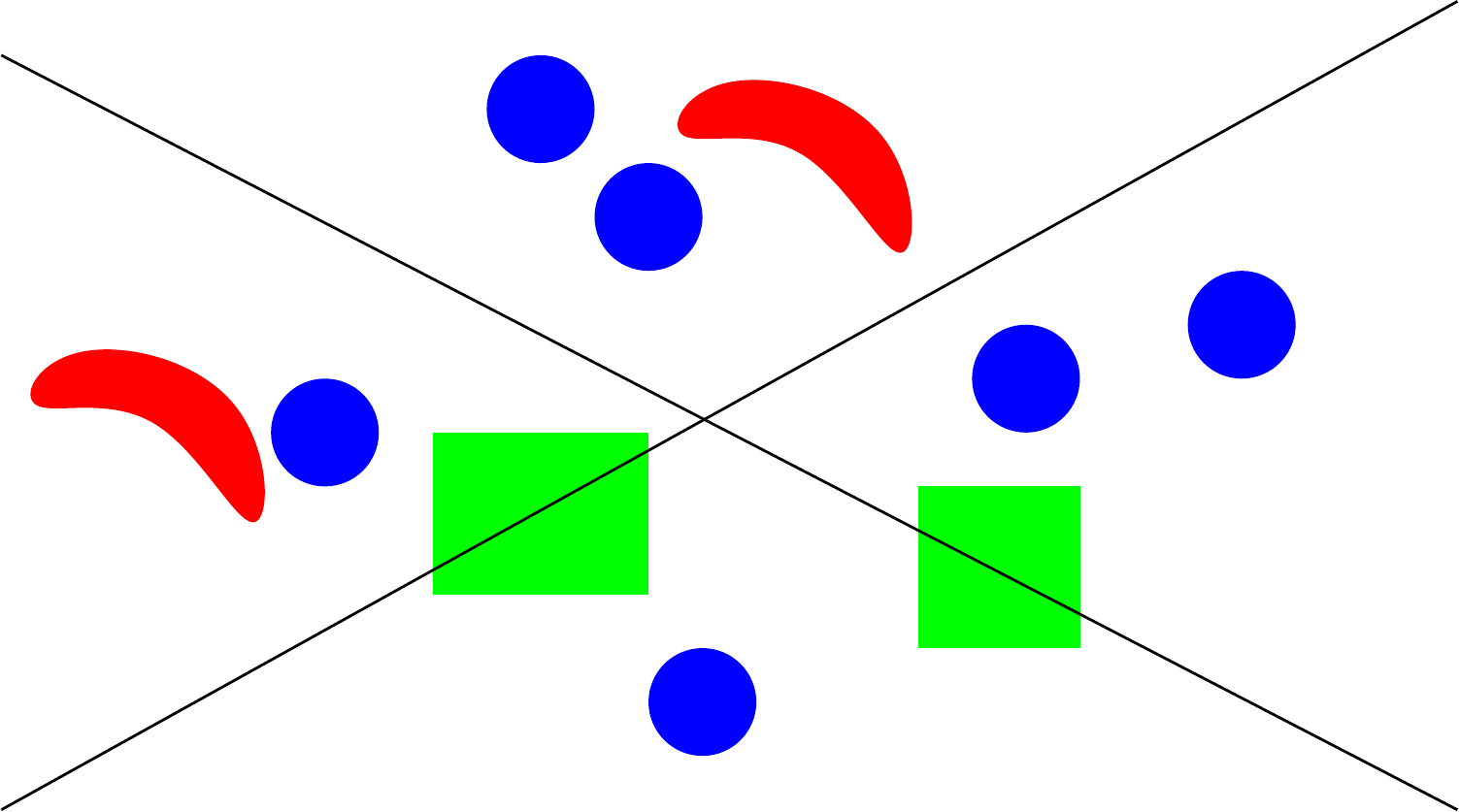}
\caption{Three masses simultaneously bisected by a double wedge.}
\label{Fig:double_wedge}
\end{figure}

It was conjectured by Langerman that any $2d$ masses in $\mathbb{R}^d$ can be simultaneously bisected by a double wedge\footnote{In fact, Langermans conjecture is even stronger, stating that any $nd$ masses in $\mathbb{R}^d$ can be simultaneously bisected by an arrangement of $n$ hyperplanes, see \cite{pizza_cccg}.}.
For $d=2$, this was confirmed by Barba, Pilz and Schnider \cite{pizza_cccg}, and the conjecture was later proven to be true for all $d$ which are a power of $2$ by Hubard and Karasev \cite{pizza1}.
For other dimensions, it is known that $2d-O(\log d)$ many masses can be simultaneously bisected by a double wedge \cite{Karasev_logd}.
Bounds for the number of masses that can be bisected by arrangements of $n$ hyperplanes were also studied in \cite{pizza2}.

Alternatively, bisections with double wedges can also be viewed as Ham-Sandwich cuts after a projective transformation: if $D=(h_1,h_2)$ is a bisecting double wedge we can find a projective transformation which sends $h_1$ to the hyperplane at infinity.
After this transformation, $h_2$ is a Ham-Sandwich cut of the transformed masses.
We extend Langermans conjecture on bisections with double wedges as follows:

\begin{conjecture}
\label{Con:projective_hs_masses}
Let $\mu_1^1,\ldots,\mu_{d+k}^1,\mu_1^2,\ldots,\mu_{d+k}^2,\ldots,\mu_{d+k}^{m}$ be $m=\lfloor\frac{d}{k}\rfloor$ families each containing $d+k$ mass distributions in $\mathbb{R}^d$.
Then there exists a projective transformation $\varphi$ such that $\varphi(\mu_1^{i}),\ldots,\varphi(\mu_{d+k}^{i})$ can be simultaneously bisected by a single hyperplane for every $i\in\{1,\ldots,m\}$.
\end{conjecture}

Langermans conjecture corresponds to the boundary case $k=d$.
In this paper, we consider the other boundary case $k=1$, which we settle for point sets instead of mass distributions:

\begin{restatable}{theorem}{projectivehs}
\label{Thm:projective_hs_points}
Let $P_1^1,\ldots,P_{d+1}^1,P_1^2,\ldots,P_{d+1}^2,\ldots,P_{d+1}^{d}$ be $d$ families each containing $d+1$ point sets in $\mathbb{R}^d$ such that their union is in general position.
Then there exists a projective transformation $\varphi$ such that $\varphi(P_1^{i}),\ldots,\varphi(P_{d+1}^{i})$ can be simultaneously bisected by a single hyperplane for every $i\in\{1,\ldots,d\}$.
\end{restatable}

We will actually prove a very similar statement for mass distributions, but in odd dimensions we need some technical restrictions.
The methods we use for this do not seem to extend to other $k>1$.
Also, the methods used for the proofs for the case $k=d$ likely do not extend to smaller $k$, so in order to settle the conjecture, new approaches might be needed.
It is plausible that successful approaches for Conjecture \ref{Con:projective_hs_masses} could also lead to new insights regarding bisections with hyperplane arrangements.

Finally, we combine fans and double wedges, to obtain the following result.

\begin{restatable}{theorem}{stripes}
\label{Thm:stripes}
Let $k=p_1p_2\cdots p_n$ be a product of pairwise distinct odd primes.
Let $2d-1\geq m(k-1)$ and let $\mu_1,\ldots,\mu_{m+1}$ be $m+1$ mass distributions in $\mathbb{R}^d$.
Then there exists a projective transformation $\varphi$ such that $\varphi(\mu_1),\ldots,\varphi(\mu_{m+1})$ can be simultaneously $(\frac{1}{k},\ldots,\frac{1}{k})$-partitioned by $k-1$ parallel hyperplanes.
\end{restatable}

\subsection{Structure of the paper}

All our results are proved using topological methods.
In some cases, we use established techniques, such as the configuration space/test map scheme (see e.g.\ the excellent book by Matou\v{s}ek \cite{Matousek}) or degree arguments.
In other cases, we use newer tools, such as a recent result about sections in canonical line bundles of flag manifolds \cite{HSSoCG}, which we rephrase as an extension of the Borsuk-Ulam theorem to flag manifolds.
Finally, some proofs rely on computing ideal-valued indexes of spaces using Leray-Serre spectral sequences.
For this reason, we do not present the proofs in the order that they appear in the introduction, but rather in order of how much background is required to understand them.

In Section \ref{sec:preliminaries}, we discuss some underlying results that appear throughout the paper, namely the Borsuk-Ulam theorem for flag manifolds and gnomonic projection.

In Section \ref{sec:cones}, we apply the Borsuk-Ulam theorem for flag manifolds to prove Theorem \ref{Thm:cones_general}.
We then prove Theorem \ref{Thm:projective_hs_points} in Section \ref{sec:double_wedges} using a combination of degree arguments and the Borsuk-Ulam theorem.
We conclude the paper by proving Theorems \ref{Thm:fans_main} and \ref{Thm:stripes} in Section \ref{sec:fans}.
The proofs in this section require different indexes of $G$-spaces and Leray-Serre spectral sequences.
We briefly recall the definitions of the indexes in the beginning of the section.
As for spectral sequences, we refer the reader to the literature, e.g.~the work by McCleary \cite{spectral}.

\section{Preliminaries}
\label{sec:preliminaries}
In this section we will discuss two technical tools that we will use throughout the paper.
The first one is a generalization of the Borsuk-Ulam theorem to flag manifolds.
This result was already used in \cite{HSSoCG}, but phrased in a slightly different setting, which is why we will still prove that the version stated here actually follows from the more general statement in \cite{HSSoCG}.
After that we give a general argument called \emph{gnomonic projection}, which will allow us to only consider wedges and cones whose apex contain the origin in the following sections.

\subsection{A Borsuk-Ulam theorem for flag manifolds}

Recall the definition of a \emph{flag manifold}: a \emph{flag} $F$ in a vector space $V$ of dimension $n$ is an increasing sequence of subspaces of the form
\[ F=\{0\}=V_0\subset V_1\subset\cdots\subset V_k=V. \]
To each flag we can assign a signature vector of dimensions of the subspaces.
A flag is a \emph{complete flag} if $\text{dim}V_i=i$ for all $i$ (and thus $k=n$).
A flag manifold $\mathcal{F}$ is the set of all flags with the same signature vector.
We denote the complete flag manifold, that is, the manifold of complete flags, by $\tilde{V}_{n,n}$.
We will consider \emph{flag manifolds with oriented lines}: a flag manifold with oriented lines $\overrightarrow{\mathcal{F}}$ is a flag manifold where each flag contains a 1-dimensional subspace (that is, a line), and this line is oriented.
We denote the complete flag manifold with oriented lines by $\overrightarrow{V}_{n,n}$.
In particular, each flag manifold with oriented lines $\overrightarrow{\mathcal{F}}$ is a double cover of the underlying flag manifold $\mathcal{F}$.
Further, reorienting the line defines a natural antipodal action.

\begin{theorem}[Borsuk-Ulam for flag manifolds]
\label{Thm:bu_flags}
Let $\overrightarrow{\mathcal{F}}$ be a flag manifold with oriented lines in $\mathbb{R}^{d+1}$. Then every antipodal map $f: \overrightarrow{\mathcal{F}}\rightarrow \mathbb{R}^d$ has a zero.
\end{theorem}

Let us briefly mention how this implies the Borsuk-Ulam theorem: the simplest flags containing a line are those of the form $F=\{0\}\subset \ell\subset\mathbb{R}^{d+1}$.
The corresponding flag manifold is the manifold of all lines through the origin in $\mathbb{R}^{d+1}$, that is, the projective space.
Orienting the lines, we retrieve the manifold of all oriented lines through the origin in $\mathbb{R}^{d+1}$, which is homeomorphic to the sphere $S^d$.
The above theorem now says that every antipodal map from $S^d$ to $\mathbb{R}^d$ has a zero, which is one of the versions of the Borsuk-Ulam theorem.

Note that any antipodal map $f: \overrightarrow{\mathcal{F}}\rightarrow \mathbb{R}^d$ can be extended to an antipodal map $g: \overrightarrow{V}_{d+1,d+1}\rightarrow \mathbb{R}^d$ by concatenating it with the natural projection from $\overrightarrow{V}_{d+1,d+1}$ to $\overrightarrow{\mathcal{F}}$, thus it is enough to show the statement for $\overrightarrow{V}_{d+1,d+1}$.

In \cite{HSSoCG} the above is phrased in terms of sections of canonical line bundles, so let us briefly recall some of the relevant terms.
A \emph{vector bundle} consists of a base space $B$, a total space $E$, and a continuous projection map $\pi: E\mapsto B$.
Furthermore, for each $b\in B$, the fiber $\pi^{-1}(b)$ over $b$ has the structure of a vector space over the real numbers.
Finally, a vector bundle satisfies the \emph{local triviality condition}, meaning that for each $b\in B$ there is a neighborhood $U\subset B$ containing $p$ such that $\pi^{-1}(U)$ is homeomorphic to $U\times\mathbb{R}^d$.
A \emph{section} of a vector bundle is a continuous mapping $s: B\mapsto E$ such that $\pi s$ equals the identity map, i.e., $s$ maps each point of $B$ to its fiber.
We can define a \emph{canonical bundle} for each $V_i$ in a complete flag, which we will denote by $\vartheta_i^d$.
The bundle $\vartheta_i^d$ has a total space $E$ consisting of all pairs $(F,v)$, where $F$ is a complete flag in $\mathbb{R}^d$ and $v$ is a vector in $V_i$, and a projection $\pi: E\mapsto \tilde{V}_{n,n}$ given by $\pi((F,v))=F$.
In \cite{HSSoCG}, the following is proved:

\begin{lemma}
\label{Lem:Sections}
Let $s_1,\ldots,s_{m+1}$ be $m+1$ sections of the canonical bundle $\vartheta_l^{m+l}$.
Then there is a flag $F\in\tilde{V}_{m+l,m+l}$ such that $s_1(F)=\ldots=s_{m+1}(F)$.
\end{lemma}

We will now show how Lemma \ref{Lem:Sections} implies Theorem \ref{Thm:bu_flags}.

\begin{proof}[Proof of Theorem \ref{Thm:bu_flags}]
Let $f: \overrightarrow{V}_{d+1,d+1}\rightarrow \mathbb{R}^d$ be antipodal.
Write $f$ as $(f_1,\ldots,f_d)$, where each $f_i: \overrightarrow{V}_{d+1,d+1}\rightarrow \mathbb{R}$ is an antipodal map.
In particular, each $f_i$ defines a section $s_i$ in $\vartheta_1^{d+1}$.
Let $s_0$ be the zero section in $\vartheta_1^{d+1}$.
Thus, $s_0,s_1,\ldots,s_d$ are $d+1$ sections in $\vartheta_1^{d+1}$ and thus by Lemma \ref{Lem:Sections} there is a flag $F$ on which the coincide.
As $s_0$ is the zero section, this means that $f_i(F)=0$ for all $i\in\{1,\ldots,d\}$ and thus $f(F)=0$, which is what we wanted to show.
\end{proof}

\subsection{Gnomonic projection}
\emph{Gnomonic projection} is a projection $\pi$ of the upper hemisphere $S^+$ of a (unit) sphere to its tangent space $T$ at the north pole.
It works as follows: for some point $p$ on $S^+$, let $\ell(p)$ be the line through $p$ and the origin.
The projection $\pi(p)$ of $p$ is then defined as the intersection of $\ell(p)$ and $T$.
Note that this is a bijection from the (open) upper hemisphere to the tangent space.
Gnomonic projection maps great circles to lines.
More generally, we say that a \emph{great $k$-circle} on a sphere $S^d$ is the intersection on $S^d$ with a $k+1$-dimensional linear subspace (i.e., a $(k+1)$-flat containing the origin).
In particular, gnomonic projection then maps great $k$-circles to $k$-flats.

Using gnomonic projection, we can show the following:

\begin{lemma}[Gnomonic projection]
\label{lem:gnomonic}
Let $k\leq d$.
Assume that any $m$ mass distributions in $\mathbb{R}^{d+1}$ can be simultaneously partitioned by a $k$-cone (or double wedge or $q$-fan) whose apex contains the origin.
Then any $m$ mass distributions in $\mathbb{R}^{d}$ can be simultaneously partitioned by a $k$-cone (or double wedge or $q$-fan).
\end{lemma}

\begin{proof}
We will only prove the statement for $k$-cones, as the proof is analogous for the other objects.
Let $\mu_1,\ldots,\mu_m$ be mass distributions in $\mathbb{R}^{d}$.
Use the inverse $\pi^{-1}$ of gnomonic projection to map $\mu_1,\ldots,\mu_m$ to the upper hemisphere of $S^d\subseteq\mathbb{R}^{d+1}$.
By our assumption there exists a $k$-cone $C$ whose apex contains the origin that simultaneously partitions $\pi^{-1}(\mu_1),\ldots,\pi^{-1}(\mu_{m})$.
Let $C_S$ be the intersection of $C$ with the upper hemisphere.
Consider now the apex of $C$, which has dimension $(d+1-k)\geq 1$.
The map $\pi$ maps this apex to a $(d-k)$ dimensional flat in $\mathbb{R}^d$.
Further, the $(d+1-k+1)$-dimensional half-hyperplanes whose union is $C$ all get mapped to $(d-k+1)$-dimensional half-hyperplanes in $\mathbb{R}^d$.
It follows that the image $\pi(C_S)$ of $C_S$ under the gnomonic projection is a $k$-cone in $\mathbb{R}^d$ which simultaneously bisects $\mu_1,\ldots,\mu_{m}$.
\end{proof}

In the following, we will only prove statements about $k$-cones, double wedges and fans whose apexes contain the origin.
The general results then follow from the above lemma.

\section{$k$-cones}\label{sec:cones}

In this section, we will use the Borsuk-Ulam theorem for flag manifolds to show the existence of bisections with $k$-cones.

\begin{theorem}
\label{Thm:cones_apex_point}
Let $\mu_1,\ldots,\mu_{d}$ be $d$ mass distributions in $\mathbb{R}^d$, let $p\in\mathbb{R}^d$ be a point and let $1\leq k\leq d$.
Then there exists a $k$-cone $C$ whose apex contains $p$ and that simultaneously bisects $\mu_1,\ldots,\mu_{d}$.
\end{theorem}

\begin{proof}
Without loss of generality, let $p$ be the origin.
Let $\overrightarrow{\mathcal{F}}$ be the flag manifold with oriented lines defined by the flags $(0,\overrightarrow{\ell},V_k,\mathbb{R}^d)$, where $\overrightarrow{\ell}$ has dimension 1 and $V_k$ has dimension $k$.
Each $(0,\overrightarrow{\ell},V_k,\mathbb{R}^d)$ defines a unique $k$-cone that bisects the total mass $\mu_1+\ldots +\mu_d$ and whose projection to $V_k$ is a cone $C$ with central axis $\overrightarrow{\ell}$ and apex $p$.
For $i\in\{1,\ldots,d-1\}$, define $f_i:=\mu_i(C)-\mu_i(\overline{C})$, where $\overline{C}$ denotes the complement of $C$.
Then $f:=(f_1,\ldots,f_{d-1})$ is a map from $\overrightarrow{\mathcal{F}}$ to $\mathbb{R}^{d-1}$.
Further, as every $C$ bisects the total mass, so does the complement $\overline{C}$.
In particular, $\overline{C}$ is the unique $k$-cone that we get when switching the orientation of $\overrightarrow{\ell}$.
Thus, $f$ is antipodal, and by Theorem \ref{Thm:bu_flags} it has a zero.
Let $C_0$ be the $k$-cone defined by this zero.
By the definition of $f$ we have that $C_0$ simultaneously bisects $\mu_1,\ldots,\mu_{d-1}$.
By construction, $C_0$ also bisects the total mass, thus, it must also bisect $\mu_d$.
\end{proof}

Theorem \ref{Thm:cones_general} now follows by gnomonic projection.

\cones*

As noted in the introduction, these results are tight with respect to the number of masses that are bisected.
However, in some cases, we can enforce additional restrictions without sacrificing a mass.

\begin{theorem}
\label{Thm:cones_apex_line}
Let $\mu_1,\ldots,\mu_{d+1}$ be $d+1$ mass distributions in $\mathbb{R}^d$, where $d$ is odd and let $g$ be a line.
Then there exists a $d$-cone $C$ whose apex $a$ lies on $g$ that simultaneously bisects $\mu_1,\ldots,\mu_{d+1}$.
\end{theorem}

\begin{proof}
Without loss of generality, let $g$ be the $x_d$-axis.
Place $a$ somewhere on $g$.
We will move $a$ along $g$ from $-\infty$ to $+\infty$.
Consider all directed lines through $a$.
Each such line $\overrightarrow{\ell}$ defines a unique $d$-cone $C$ bisecting the total mass.
For each $i\in\{1,\ldots,d\}$, define $f_i:=\mu_i(C)-\mu_i(\overline{C})$, where $\overline{C}$ again denotes the complement of $C$.
Thus, for each choice of $a$, we get a map $f_a=(f_1,\ldots,f_d): S^{d-1}\rightarrow\mathbb{R}^d$.
We claim that for some $a$ this map has a zero.
Assume for the sake of contradiction that none of the maps have a zero.
Then we can normalize them to get maps $f_a: S^{d-1}\rightarrow S^{d-1}$.
In particular, each such map has a degree.
Further, all of the maps are antipodal, implying that their degree is odd, and thus non-zero.
Finally, note that $f_{-\infty}=-f_{+\infty}$, and thus $\deg(f_{-\infty})=-\deg(f_{+\infty})$.
(Here we require that $d-1$ is even.)
In particular, as the degrees are non-zero, $f_{-\infty}$ and $f_{+\infty}$ have different degrees.
But moving $a$ along $g$ from $-\infty$ to $+\infty$ defines a homotopy from $f_{-\infty}$ to $f_{+\infty}$, which is a contradiction.
Thus, there exists some $a$ such that $f_a$ has a zero and analogous to above this zero defines a $d$-cone $C$ that simultaneously bisects $\mu_1,\ldots,\mu_{d+1}$.
\end{proof}

\section{Double Wedges}\label{sec:double_wedges}

Before proving Theorem \ref{Thm:projective_hs_points}, we will prove a more general statement about bisections of mass distributions with double wedges.
Let us first explain how bisections with double wedges can be regarded as Ham-Sandwich cuts after a projective transformation:
Let $\mu$ be a mass distribution in $\mathbb{R}^d$ and let $D=(h_1,h_2)$ be a double wedge that bisects $\mu$.
Use gnomonic projection to map $\mu$ and $D$ to the upper hemisphere of $S^d\subseteq\mathbb{R}^{d+1}$.
Now, antipodally copy $\mu$ and $D$ to the lower hemisphere.
Note that both $h_1$ and $h_2$ are oriented $(d-1)$-dimensional great circles on $S^d$, so we can extend them to oriented hyperplanes through the origin in $\mathbb{R}^{d+1}$, which we denote as $H_1$ and $H_2$, respectively.
Also, we will denote the defined measure on $S^d$ by $\mu_S$.
Note now that $\mu_S(S^d)=2\mu(\mathbb{R}^d)$ and that $(H_1,H_2)$ bisects $\mu_S$.
Further, the above is invariant under rotations of the sphere, thus we can rotate the sphere until $H_1$ is one of the two orientations of the hyperplane $x_d=0$.
Using gnomonic projection to map the upper hemisphere to $\mathbb{R}^d$, we get a projective transformation $\varphi$ of $\mathbb{R}^d$ with the property that $\varphi(h_1)$ is the sphere at infinity and that $\varphi(h_2)$ bisects $\varphi(\mu)$.
Thus, we get the following:

\begin{lemma}
\label{Lem:dw_projective}
Let $\mu_1,\ldots,\mu_k$ be mass distributions in $\mathbb{R}^d$ and let $D=(h_1,h_2)$ be a double wedge which simultaneously bisects $\mu_1,\ldots,\mu_k$.
Then there is a projective transformation $\varphi$ of $\mathbb{R}^d$ with the property that $\varphi(h_1)$ is the sphere at infinity and that $\varphi(h_2)$ simultaneously bisects $\varphi(\mu_1),\ldots,\varphi(\mu_k)$.
\end{lemma}

In the following we will now prove results about bisections of different families with different double wedges which still share one of the hyperplanes.

\begin{lemma}
\label{Lem:dw_masses_odd}
Let $\mu_1^1,\ldots,\mu_{d}^1,\mu_1^2,\ldots,\mu_{d}^2,\ldots,\mu_{d}^{d-1}$ be $d-1$ families each containing $d$ mass distributions in $\mathbb{R}^d$, where $d$ is odd.
Then there exists an oriented hyperplane $h_1$ and $d-1$ double wedges $D^i=(h_1,h_2^i)$, $i\in\{1,\ldots,d-1\}$, whose apexes all contain the origin and such that $D^i$ simultaneously bisects $\mu_1^i,\ldots,\mu_{d}^i$.
\end{lemma}

\begin{proof}
The space of pairs $(h_1,h_2)$ of oriented hyperplanes in $\mathbb{R}^d$ containing the origin is $S^{d-1}\times S^{d-1}$.
For some mass distribution $\mu$, we can thus define a function $f: S^{d-1}\times S^{d-1}\rightarrow \mathbb{R}$ by $f(h_1,h_2):=\mu(D)-\mu(\overline{D})$, where $D$ is the double wedge defined by $h_1$ and $h_2$.
Note that $f(-h_1,h_2)=f(h_1,-h_2)=-f(h_1,h_2)$ and that $f(h_1,h_2)=0$ if and only if $D$ bisects $\mu$.
Further, for each $h_1\in S^{d-1}$ we get a function $f_{h_1}:=f(h_1,\cdot): S^{d-1}\rightarrow \mathbb{R}$.

Let us now fix some $h_1$ and consider the family $\mu_1^1,\ldots,\mu_{d}^1$.
Assume that $\mu_1^1,\ldots,\mu_{d}^1$ cannot be simultaneously bisected by a double wedge defined by $h_1$ and some other hyperplane through the origin.
In particular, defining a function as above for each mass yields a map $g_{h_1}:S^{d-1}\rightarrow\mathbb{R}^d$ which has no zero.
Thus, after normalizing, we get a map $g_{h_1}:S^{d-1}\rightarrow S^{d-1}$.
In particular, this map has a degree.
As we have $g_{h_1}(-h_2)=-g_{h_1}(h_2)$, this degree is odd and thus non-zero.
Further, varying $h_1$ again, we note that $g_{-h_1}(h_2)=-g_{h_1}(h_2)$, and thus, as $d-1$ is even, we have $\deg(g_{-h_1})=-\deg(g_{h_1})$.
In particular, any path from $-h_1$ to $h_1$ defines a homotopy between two maps of different degree, which is a contradiction.
Thus, along every path from $-h_1$ to $h_1$ we encounter a hyperplane $h_1^*$ such that $g_{h_1^*}$ has a zero.
In particular, this partitions $S^{d-1}$ into regions where $g_{h_1}$ has a zero and where it does not.
Let $Z\subseteq S^{d-1}$ be the region where $g_{h_1}$ has a zero.
We note that all regions are antipodal (i.e., $g_{h_1}$ has a zero if and only if $g_{-h_1}$ does) and no connected component of $S^{d-1}\setminus Z$ contains two antipodal points.
Hence, we can define a map $t_1: S^{d-1}\rightarrow\mathbb{R}$ as follows: for each $h_1\in Z$, set $t_1(h_1)=0$.
Further, for each $h_1\in S^{d-1}\setminus Z$, set $t_1(h_1)=\deg(g_{h_1})\cdot d(h_1,Z)$, where $d(h_1,Z)$ denotes the distance from $h_1$ to $Z$.
Note that $t_1$ is continuous and $t_1(-h_1)=-t_1(h_1)$ and $t_1(h_1)=0$ if and only if $g_{h_1}$ has a zero.

We can do this for all families to get an antipodal map $t:=(t_1,\ldots,t_{d-1}):S^{d-1}\rightarrow\mathbb{R}^{d-1}$.
By the Borsuk-Ulam theorem, this map has a zero.
This zero gives us a hyperplane $h_1$ through the origin which, by construction, has the property that for each family of masses $\mu_1^i,\ldots,\mu_{d}^i$ there exists another hyperplane $h_2$ through the origin such that $(h_1,h_2)$ simultaneously bisects $\mu_1^i,\ldots,\mu_{d}^i$.
\end{proof}

Using the same lifting argument as for the proof of Theorem \ref{Thm:cones_general}, we get the following:

\begin{corollary}
\label{Cor:dw_masses_even}
Let $\mu_1^1,\ldots,\mu_{d+1}^1,\mu_1^2,\ldots,\mu_{d+1}^2,\ldots,\mu_{d+1}^d$ be $d$ families each containing $d+1$ mass distributions in $\mathbb{R}^d$, where $d$ is even.
Then there exists an oriented hyperplane $h_1$ and $d$ double wedges $D^i=(h_1,h_2^i)$, $i\in\{1,\ldots,d\}$, such that $D^i$ simultaneously bisects $\mu_1^i,\ldots,\mu_{d+1}^i$.
\end{corollary}

Note that for the argument $\deg(g_{-h_1})=-\deg(g_{h_1})$ we require that the considered sphere has even dimension.
This means, that if we want to prove Lemma \ref{Lem:dw_masses_odd} for even dimensions, we have to use different arguments.
In the following we try to do this, but at the expense that we will only be able to 'almost' bisect the masses.
More precisely, we say that a double wedge $D=(h_1,h_2)$ \emph{$\varepsilon$-bisects} a mass $\mu$ if $|\mu(D)-\mu(\overline{D})|<\varepsilon$.
Similarly we say that $D$ simultaneously $\varepsilon$-bisects $\mu_1,\ldots,\mu_k$ if it $\varepsilon$-bisects $\mu_i$ for every $i\in\{1,\ldots,k\}$.
In the following we will show Lemma \ref{Lem:dw_masses_odd} for even dimensions, with 'bisect' replaced by '$\varepsilon$-bisect'.
For this we first need an auxiliary lemma.

\begin{lemma}
\label{Lem:dw_masses_odd_point}
Let $\delta>0$ and let $\mu_1^1,\ldots,\mu_{d-1}^1,\mu_1^2,\ldots,\mu_{d-1}^2,\ldots,\mu_{d-1}^{d-2}$ be $d-2$ families each containing $d-1$ mass distributions in $\mathbb{R}^d$, where $d$ is odd.
Then there exists an oriented hyperplane $h_1$ containing the $x_d$-axis, a point $p$ at distance $\delta$ to the $x_d$-axis and $d-2$ double wedges $D^i=(h_1,h_2^i)$, $i\in\{1,\ldots,d-1\}$, whose apexes all contain the origin and such that $D^i$ simultaneously bisects $\mu_1^i,\ldots,\mu_{d-1}^i$ and $h_2^i$ contains $p$.
\end{lemma}

The proof is very similar to the proof of Lemma \ref{Lem:dw_masses_odd}.

\begin{proof}
The space of pairs $(h_1,h_2)$ of oriented hyperplanes in $\mathbb{R}^d$ with $h_1$ containing the $x_d$-axis and $h_2$ containing the origin is $S^{d-2}\times S^{d-1}$.
Let us again fix some $h_1$ and consider the family $\mu_1^1,\ldots,\mu_{d-1}^1$.
As above, we can define $d-1$ functions which give rise to a function $g_{h_1}:S^{d-1}\rightarrow\mathbb{R}^{d-1}$, which has a zero if and only if $\mu_1^1,\ldots,\mu_{d-1}^1$ can be simultaneously bisected by a double wedge using $h_1$.
Further, for each $h_1$ we can define in a continuous fashion a point $p$ which lies in the positive side of $h_1$ and on the upper hemisphere of $S^{d-1}\subseteq \mathbb{R}^d$, and which has distance $\delta$ to the $x_d$-axis.
Define now $d_{h_1}:=I_{p\in D}\cdot d(p,h_2)$, where $d(p,h_2)$ denotes the distance from $p$ to $h_2$ and $I_{p\in D}=1$ if $p$ lies in the double wedge $D=(h_1,h_2)$ and $I_{p\in D}=-1$ otherwise.
Note that $d_{h_1}$ is continuous, $d_{h_1}=0$ if and only if $p$ is on $h_2$, and $d_{h_1}(-h_2)=-d_{h_1}(h_2)$.
Thus, together with the $d-1$ functions defined by the masses, we get a map $g_{h_1}:S^{d-1}\rightarrow\mathbb{R}^{d}$, which has a zero if and only if $\mu_1^1,\ldots,\mu_{d-1}^1$ can be simultaneously bisected by a double wedge using $h_1$ and an $h_2$ passing through $p$.
Again, assuming this map has no zero, we get a map $g_{h_1}:S^{d-1}\rightarrow S^{d-1}$, which, because of the antipodality condition, has odd degree.
Again, we have $g_{-h_1}(h_2)=-g_{h_1}(h_2)$, and thus along every path from $-h_1$ to $h_1$ we encounter a hyperplane $h_1^*$ such that $g_{h_1^*}$ has a zero.
As above, this partitions the sphere into antipodal regions, the only difference being that this time we only consider the sphere $S^{d-2}$.
In particular, the Borsuk-Ulam theorem now gives us a hyperplane $h_1$ containing the $x_d$-axis which, by construction, has the property that for each family of masses $\mu_1^i,\ldots,\mu_{d-1}^i$ there exists another hyperplane $h_2^i$ through the origin such that $(h_1,h_2^i)$ simultaneously bisects $\mu_1^i,\ldots,\mu_{d-1}^i$.
Further, the hyperplane $h_1$ also defines a point $p$ at distance $\delta$ to the $x_d$-axis with the property that each $h_2^i$ contains $p$.
\end{proof}

After gnomonic projection, we thus get a double wedges $(h_1,h_2^i)$ that simultaneously bisect the masses $\mu_1^i,\ldots,\mu_d^i$ and such that $h_1$ contains the origin and the distance from $h_2^i$ is at most $\delta$.
If we now translate $h_2^i$ to contain the origin, by continuity we get that there is some $\varepsilon>0$ such that $(h_1,h_2^i)$ simultaneously $\varepsilon$-bisects the masses $\mu_1^i,\ldots,\mu_d^i$.
In particular, for every $\varepsilon>0$ we can choose $\delta>0$ in the above lemma such that after gnomonic projection we get the following:

\begin{corollary}
\label{Cor:eps_bisections}
Let $\varepsilon>0$.
\begin{enumerate}
\item  Let $\mu_1^1,\ldots,\mu_{d}^1,\mu_1^2,\ldots,\mu_{d}^2,\ldots,\mu_{d}^{d-1}$ be $d-1$ families each containing $d$ mass distributions in $\mathbb{R}^d$, where $d$ is even.
Then there exists an oriented hyperplane $h_1$ containing the origin and $d-1$ double wedges $D^i=(h_1,h_2^i)$, $i\in\{1,\ldots,d-1\}$, whose apexes all contain the origin and such that $D^i$ simultaneously $\varepsilon$-bisects $\mu_1^i,\ldots,\mu_{d}^i$.
\item Let $\mu_1^1,\ldots,\mu_{d+1}^1,\mu_1^2,\ldots,\mu_{d+1}^2,\ldots,\mu_{d+1}^d$ be $d$ families each containing $d+1$ mass distributions in $\mathbb{R}^d$, where $d$ is odd.
Then there exists an oriented hyperplane $h_1$ and $d$ double wedges $D^i=(h_1,h_2^i)$, $i\in\{1,\ldots,d\}$, such that $D^i$ simultaneously $\varepsilon$-bisects $\mu_1^i,\ldots,\mu_{d+1}^i$.
\end{enumerate}
\end{corollary}

By a standard argument (see e.g. \cite{Matousek}), a bisection partition result for mass distributions also implies the analogous result for point sets in general position.
Further, for point sets in general position, we can choose $\varepsilon$ small enough to get an actual bisection.
Thus, Theorem \ref{Thm:projective_hs_points} now follows from Lemma \ref{Lem:dw_projective}, Corollary \ref{Cor:dw_masses_even} and the second part of Corollary \ref{Cor:eps_bisections}.

\projectivehs*

Let us mention that this result is tight with respect to the number of families that are bisected:
Consider $d+1$ families each containing $d+1$ point sets in $\mathbb{R}^d$ where each point set consists of many points that are very close together.
Place these point sets in such a way that no hyperplane passes through $d+1$ of them.
Look at one family of $d+1$ point sets.
if $(h_1,h_2)$ are to simultaneously bisect the point sets in this family, each point set must be transversed by either $h_1$ or $h_2$, or both.
In particular, as $h_2$ can pass through at most $d$ point sets, $h_1$ must pass through at least one point set.
This is true for all $d+1$ families, which means that $h_1$ must pass through at least $d+1$ point sets in total, which cannot happen by our construction of the point sets.

For larger families, a similar argument shows that at most $\lfloor\frac{d}{k}\rfloor$ families each containing $d+k$ point sets can have a Ham-Sandwich cut after a common projective transformation, showing that Conjecture \ref{Con:projective_hs_masses}, if true, would be tight.
%


\section{Fans}
\label{sec:fans}

Similar to the previous sections, we will again first prove all results with the apex containing the origin.
The general results will then again follow from a lifting argument.
Our proofs are very similar to those in \cite{Barany2}.

The main idea is the following: consider the Stiefel manifold $V_{2}(\mathbb{R}^d)$ of all pairs $(x,y)$ of orthonormal vectors in $\mathbb{R}^d$.
Each such pair defines an oriented plane $h$ which contains $x$ and $y$.
For each pair $(x,y)$ we can define a $k$-fan through the origin whose apex is the orthogonal complement of $h$ and for which the first semi-hyperplane is spanned by the apex and the vector $x$.
The other semi-hyperplanes are then constructed by rotating the semi-hyperplanes around the apex in the direction of $y$ such that the resulting $k$-fan equipartitions the first given mass $\mu_1$.
This gives a unique $k$-fan with a rotational order on the wedges $W_1,\ldots,W_k$.
Further, there is a natural $\mathbb{Z}_k$-action sending $W_i$ to $W_{i+1}$ (and $W_k$ to $W_1$), giving $V_{2}(\mathbb{R}^d)$ the structure of a $\mathbb{Z}_k$-space.

For all other masses $\mu_2,\ldots,\mu_m$, we consider the map $f_i:=(1/k-\mu_i(W_1),\ldots,1/k-\mu_i(W_k)$.
This map is zero if and only if the $k$-fan equipartitions the mass $\mu_i$.
Let $Z$ be the target space of $f=(f_2,\ldots,f_m)$.
The $\mathbb{Z}_k$-action on the fans induces a $\mathbb{Z}_k$-action on $Z$.
Simultaneous equipartition of all the masses then follows from the non-existence of a $\mathbb{Z}_k$-equivariant map from $V_{2}(\mathbb{R}^d)$ to $Z\setminus\{0\}$.

In the following, we will investigate some conditions under which such a map cannot exist.
For this, we first introduce some terminology.

Recall the following indexes of $G$-spaces, due to Fadell and Husseini \cite{Fadell} as well as Volovikov \cite{Volovikov}:
Let $EG\rightarrow BG$ be the universal $G$-bundle.
The key property of the $G$-bundle is that if $X$ is a $CW$-complex on which $G$ acts freely (e.g. $V_{2,d}$) then for its associated $G$-bundle $X\rightarrow X/G$ there is a (homotopically) unique map $\alpha_X: X/G\rightarrow BG$ such that the bundle $X\rightarrow X/G$ is isomorphic to the bundle obtained by pulling back the bundle $EG\rightarrow BG$ along $\alpha_X$.

More generally, consider the Borel fibration $X\rightarrow X_G\rightarrow BG$, where $X_G=(EG\times X)/G$.
Define the \emph{cohomological index} $\text{Ind}_G(X)$ as the kernel of the homomorphism $p_X^*: H^*(BG;R)\rightarrow H^*(X_G;R)$.
The cohomological Leray-Serre spectral sequence of the Borel fibration converges to $H^*(X_G;R)$ (and in case $X$ is free to $H^*(X/G;R)$) and we have $E_2^{p,q}\cong H^p(BG;\mathcal{H}^q(X;R))$, where $\mathcal{H}^q(X;R)$ are local coefficients.
Further, $p_x^*$ is the composition $H^*(BG;R)\rightarrow E_2^{*,0}\rightarrow E_3^{*,0}\rightarrow\ldots\rightarrow E_{\infty}^{*,0}\subseteq H^*(X_G;R)$.
We now denote by $_r\text{Ind}_G(X)$ the kernel of $H^*(BG;R)\rightarrow E_{r+1}^{*,0}$.
Note that we have $_1\text{Ind}_G(X)\subseteq _2\text{Ind}_G(X)\subseteq\ldots$ and that the union of these ideals is $\text{Ind}_G(X)$.

Let now $0\neq\alpha\in\text{Ind}_G(X)$ and let $i_X(\alpha)$ be the smallest integer $r$ such that $\alpha\in _r\text{Ind}_G(X)$.
We define the \emph{numerical index} $i_G(X)$ as the minimum of $i_X(\alpha)$ taken over all non-zero $\alpha$ in $\text{Ind}_G(X)$.
If $\text{Ind}_G(X)=0$, we set $i_G(X)=\infty$.
In other words, $i_G(X)$ is the index of the first differential in the spectral sequence that kills an element of $H^*(X_G;R)$.

The indexes defined above can be used to show the non-existence of certain $G$-equivariant maps:

\begin{lemma}[\cite{Fadell, Volovikov}]
\label{lem:index}
Let $X$ and $Y$ be $G$-spaces and let $f:X\rightarrow Y$ be a $G$-equivariant map.
Then
\begin{enumerate}
\item $\text{Ind}_G(X)\supseteq\text{Ind}_G(Y)$ and
\item $i_G(X)\leq i_G(Y)$.
\end{enumerate}
\end{lemma}

In order to prove that for some choices of $k$ there is no $\mathbb{Z}_k$-equivariant map from $V_{2}(\mathbb{R}^d)$ to $Z\setminus\{0\}$, it is thus sufficient to compute the $\mathbb{Z}_k$-indexes of $V_{2}(\mathbb{R}^d)$ and $Z\setminus\{0\}$.
This is what we will do in the following.

We first introduce two lemmas which are helpful for computing the indexes.

\begin{lemma}[\cite{Volovikov}, Proposition 2.5]
\label{lem:numerical_index}
Let $X$ and $Y$ be $G$-spaces.
\begin{enumerate}
\item If $\tilde{H}^i(X;R)=0$ for $i<n$ then $i_G(X)\geq n+1$.
\item If $H^i(Y;R)=0$ for $i>n-1$ and $i_G(Y)<\infty$ (or, equivalently, $\text{Ind}_G(Y)\neq 0$) then $i_G(Y)\leq n$.
\end{enumerate}
\end{lemma}

In particular, if $Y$ has dimension at most $n-1$ and $G$ acts freely on $Y$, then $\text{Ind}_G(Y)$ is the kernel of the map $H^*(BG;R)\rightarrow H^*(Y/G)$ (\cite{Fadell}, Remark 3.15) and thus $H^i(BG;R)\subseteq\text{Ind}_G(Y)$ for all $i>n-1$.
Hence, assuming that for some $k>n-1$ we have $H^k(BG;R)\neq 0$, it follows that $i_G(Y)\leq n$.

Further, we will need two results about indexes of products of groups and of union of spaces.

\begin{lemma}[\cite{Fadell}, Corollary 3.4]
\label{lem:index_product}
Let $X$ be a $G_1\times G_2$-space, where $G_1\times G_2$ acts on $X$ by $(g_1,g_2)x=g_1x$.
Then $\text{Ind}_{G_1\times G_2}(X)=(\text{Ind}_{G_1}(X))\otimes H^*(BG_2)$.
\end{lemma}

\begin{lemma}[\cite{Volovikov}, Proposition 2.2]
\label{lem:index_union}
Let $X$ be the union of two closed or two open invariant subspaces $A$ and $B$.
Then $\text{Ind}_{G}(A)\cdot\text{Ind}_{G}(B)\subseteq\text{Ind}_{G}(X)$.
\end{lemma}

The index of $\mathbb{Z}_p$-actions on $V_{2}(\mathbb{R}^d)$ for odd primes $p$ was already studied by Makeev \cite{Makeev}.
As there does not seem to be an English version of his proof available, we include a proof here.
The ideas behind the proof presented here were communicated to the author by Roman Karasev \cite{Roman}.
We present a slightly different proof here (using a different index) which, as we will see, allows a generalization to all finite cyclic groups of odd order.

\begin{lemma}[Makeev, communicated by Karasev]
\label{lem:makeev}
Let $p$ be an odd prime.
Then $i_{\mathbb{Z}_p}(V_{2}(\mathbb{R}^d))=2d-2$.
\end{lemma}

\begin{proof}
We work with cohomology over the integers.
Then, the cohomology groups of $V_{2}(\mathbb{R}^d)$ are as follows  (see \cite{Borel}, Prop. 10.1):
If $d$ is odd, we have $H^i(V_{2}(\mathbb{R}^d);\mathbb{Z})\cong\mathbb{Z}$ for $i\in\{0,2d-3\}$, $H^{d-1}(V_{2}(\mathbb{R}^d);\mathbb{Z})\cong\mathbb{Z}_2$, and all other cohomology groups are trivial.
If $d$ is even, we have $H^i(V_{2}(\mathbb{R}^d);\mathbb{Z})\cong\mathbb{Z}$ for $i\in\{0,d-2,d-1,2d-3\}$ and all other cohomology groups are trivial.
Further, from \cite{tomDieck}, III.2 we get that the cohomology groups of the space $B\mathbb{Z}_p$ are
\[   
H^i(B\mathbb{Z}_p;\mathbb{Z}) = 
     \begin{cases}
       0&\quad i\equiv 1\text{ mod }2\\
       \mathbb{Z}_p&\quad i\equiv 0\text{ mod }2, i\neq 0\\
       \mathbb{Z}&\quad i=0\\
     \end{cases}
\]

Recall that in the cohomological Leray-Serre we have $E_2^{p,q}\cong H^p(BG;\mathcal{H}^q(X;R))$, where $\mathcal{H}^q(X;R)$ are local coefficients.
As $X=V_{2,d}$ is path connected and simply connected, we have $E_2^{p,q}\cong H^p(B\mathbb{Z}_p;\mathcal{H}^q(V_{2}(\mathbb{R}^d);\mathbb{Z}))\cong H^p(B\mathbb{Z}_p;\mathbb{Z})\otimes H^q(V_{2}(\mathbb{R}^d);\mathbb{Z})$.

Let us first assume that $d$ is odd.
Then the second page of the spectral sequence only has 3 non-trivial rows: $E_2^{*,0}$, $E_2^{*,d-1}$ and $E_2^{*,2d-3}$.
In particular, the only potentially non-trivial differentials are $d_{d}$ and $d_{2d-2}$.
The differential $d_d$ must be trivial as $H^{d-1}(V_{2}(\mathbb{R}^d);\mathbb{Z})\cong\mathbb{Z}_2$, whereas $p$ is odd.
Thus, the only potentially non-trivial differential is $d_{2d-2}$, showing that $i_{\mathbb{Z}_p}(V_{2}(\mathbb{R}^d))\geq 2d-2$.
On the other hand, the $\mathbb{Z}_p$-action on $V_{2}(\mathbb{R}^d)$ is free and $V_{2}(\mathbb{R}^d)$ is a manifold of dimension $2d-3$, thus by Lemma \ref{lem:numerical_index} we get $i_{\mathbb{Z}_p}(V_{2}(\mathbb{R}^d))\leq 2d-2$, which finishes the proof for the case where $d$ is odd.

If $d$ is even, the second page of the spectral sequence has 4 non-trivial rows: $E_2^{*,0}$, $E_2^{*,d-2}$, $E_2^{*,d-1}$ and $E_2^{*,2d-3}$, implying that $i_{\mathbb{Z}_p}(V_{2}(\mathbb{R}^d))\in\{d-1, d, 2d-2\}$.
As the embedding $V_{2}(\mathbb{R}^{d-1})\hookrightarrow V_{2}(\mathbb{R}^{d})$ is $\mathbb{Z}_p$-equivariant, it follows from Lemma \ref{lem:index} that $i_{\mathbb{Z}_p}(V_{2}(\mathbb{R}^d))\geq i_{\mathbb{Z}_p}(V_{2}(\mathbb{R}^{d-1}))= 2d-4$, which proves the claim for all $d\neq 4$.

It remains to show that $i_{\mathbb{Z}_p}(V_{2}(\mathbb{R}^4))=6$.
From the above arguments we get $i_{\mathbb{Z}_p}(V_{2}(\mathbb{R}^d))\in\{4,6\}$.
If $i_{\mathbb{Z}_p}(V_{2}(\mathbb{R}^4))=4=i_{\mathbb{Z}_p}(V_{2}(\mathbb{R}^3))$, i.e., the differential $d_4$ is not trivial, we would get a non-trivial map $H^3(V_{2}(\mathbb{R}^4))\to H^3(V_{2}(\mathbb{R}^3))$ from the spectral sequence.
However, such a map cannot exist as the composition $V_{2}(\mathbb{R}^3)\to V_{2}(\mathbb{R}^4)\to S^3$, where the second map sends $(x,y)$ to $x$, is null-homotopic.
\end{proof}

We remark here that we could also have shown the above using cohomology over $\mathbb{Z}_p$.
Also, we could have computed the ideal-valued cohomological index $\text{Ind}_{\mathbb{Z}_p}$ instead of the numerical index $i_{\mathbb{Z}_p}$.
For the case where $d$ is odd, this was done by Jelic \cite{Jelic}.
Her computations also go through analogously for $d$ even (using the embedding $V_{2}(\mathbb{R}^{d-1})\hookrightarrow V_{2}(\mathbb{R}^{d})$) if we consider $_{2d-2}\text{Ind}_{\mathbb{Z}_p}$ instead of $\text{Ind}_{\mathbb{Z}_p}$.

Note that in the above proof we only used the fact that $p$ is odd, but never the fact that $p$ is prime.
Thus, what we have actually shown is the following:
\begin{corollary}
\label{cor:makeev_general}
Let $k$ be an odd integer.
Then $i_{\mathbb{Z}_k}(V_{2}(\mathbb{R}^d))=2d-2$.
\end{corollary}

Before we give a more general result, for the sake of completeness we first show a proof of Makeev's result on partitions with $p$-fans.

\begin{theorem}
\label{Thm:MakeevOrigin}
Let $p$ be an odd prime and let $2d-3\geq m(p-1)$.
Then any $m+1$ mass distributions $\mathbb{R}^d$ can be simultaneously $(\frac{1}{p}\ldots,\frac{1}{p})$-partitioned by a $p$-fan through the origin.
\end{theorem}

Makeev's result now follows from gnomonic projection (Lemma \ref{lem:gnomonic}).

\begin{corollary}[Makeev, see \cite{KarasevSurvey} Theorem 57]
Let $p$ be an odd prime and let $2d-1\geq m(p-1)$.
Then any $m+1$ mass distributions $\mathbb{R}^d$ can be simultaneously $(\frac{1}{p}\ldots,\frac{1}{p})$-partitioned by a $p$-fan.
\end{corollary}

\begin{proof}[Proof of Theorem \ref{Thm:MakeevOrigin}]
Assume without loss of generality that for each mass distribution $\mu_i$ we have $\mu_i(\mathbb{R^d})=1$.
Let $m$ be such that $2d-1\geq m(p-1)$.
Consider the Stiefel manifold $V_{2}(\mathbb{R}^d)$ of all pairs $(x,y)$ of orthonormal vectors in $\mathbb{R}^d$.
To each $(x,y)\in V_2(\mathbb{R}^d)$ we assign a $p$-fan $F(x,y)$ as follows:
Let $h$ by the linear subspace spanned by $(x,y)$ and let $\pi:\mathbb{R}^d\rightarrow h$ be the canonical projection.
The apex of the $p$-fan $F(x,y)$ is then $\pi^{-1}(0)$.
Further, note that $(x,y)$ defines an orientation on $h$, so we can consider a ray on $h$ rotating in clockwise direction.
Start this rotation at $x$, and let $r_1$ be the (unique) ray such that the area between $x$ and $r_1$ is the projection of a wedge $W_1$ which contains exactly a $\frac{1}{p}$-fraction of the total mass.
Analogously, let $r_i$ be the (unique) ray such that the area between $r_{i-1}$ and $r_i$ is the projection of a wedge $W_i$ which contains exactly a $\frac{1}{p}$-fraction of the total mass.
This construction thus continuously defines a $p$-fan $F(x,y)$ through the origin for each $(x,y)\in V_2(\mathbb{R}^d)$.
Further note that there is a natural $\mathbb{Z}_p$ action on $V_{2}(\mathbb{R}^d)$, defined by $(W_1,W_2,\ldots,W_p)\mapsto (W_p,W_1,\ldots,W_{p-1})$, i.e., by turning by one sector.

For a mass distribution $\mu_i$, $i\in\{1,\ldots,m\}$ we introduce a test map $f_i(x,y):V_{2}(\mathbb{R}^d)\rightarrow\mathbb{R}^p$ by
$$f_i(x,y):=(\mu_i(W_1)-\frac{1}{p},\mu_i(W_2)-\frac{1}{p},\ldots,\mu_i(W_p)-\frac{1}{p}).$$
Note that the image of $f_i$ is contained in the hyperplane $Z=\{y\in\mathbb{R^p}:y_1+y_2+\ldots+y_p=0\}$ of dimension $p-1$ and that $f_i(x,y)=0$ implies that $F(x,y)$ $(\frac{1}{p}\ldots,\frac{1}{p})$-partitions $\mu_i$.
In particular, if $f_i(x,y)=0$ for all $i\in\{1,\ldots,m\}$, then $F(x,y)$ simultaneously equipartitions $\mu_1,\ldots,\mu_m$, and thus, as $F(x,y)$ equipartitions the total mass by construction, it also equipartitions $\mu_{m+1}$.
We thus want to show that all test maps have a common zero.
To this end, we note that the $\mathbb{Z}_p$ action on $V_{2}(\mathbb{R}^d)$ induces a $\mathbb{Z}_p$ action $\nu$ on $Z$ by $\nu(y_1,y_2,\ldots,y_p)=(y_p,y_1,\ldots,y_{p-1})$.
Further, as $p$ is prime, this action is free on $Z\setminus\{0\}$ and thus also on $Z^m\setminus\{0\}$.
Thus, if we assume that the test maps do not have a common zero, they induce a $\mathbb{Z}_p$-map $f:V_{2}(\mathbb{R}^d)\rightarrow Z^m\setminus\{0\}$.
We will now show that there is no such map.

To this end, we first note that the dimension of $Z^m$ is $m(p-1)$ and that, after normalizing, $f$ induces a map $f': V_{2}(\mathbb{R}^d)\rightarrow S^{m(p-1)-1}$.
As $\mathbb{Z}_p$ acts freely on $Z^m\setminus\{0\}$, it also does so on $S^{m(p-1)-1}$, thus we get from Lemma \ref{lem:numerical_index} that $i_{\mathbb{Z}_p}(S^{m(p-1)-1})\leq\text{dim}(S^{m(p-1)-1})+1=m(p-1)$.
On the other hand, from Lemma \ref{lem:makeev} we have $i_{\mathbb{Z}_p}(V_{2}(\mathbb{R}^d))=2d-2$.
By Lemma \ref{lem:index}, $f'$ (and thus also $f$) can only exist if $2d-2\leq m(p-1)$.
But $m$ was chosen such that $2d-3\geq m(p-1)$, so $f'$ and $f$ cannot exist.
\end{proof}

Using the techniques from \cite{Barany2}, this can be extended to more general partitions.

\begin{theorem}
Let $p$ be an odd prime, let $(a_1,\ldots,a_q)\in\mathbb{N}^q$ with $q<p$ and $a_1+\ldots +a_q=p$
\begin{itemize}
\item Let $2d-2\geq m(p-1)$.
Then any $m+1$ mass distributions $\mathbb{R}^d$ can be simultaneously $(\frac{a_1}{p}\ldots,\frac{a_q}{p})$-partitioned by a $q$-fan through the origin.
\item Let $2d\geq m(p-1)$.
Then any $m+1$ mass distributions $\mathbb{R}^d$ can be simultaneously $(\frac{a_1}{p}\ldots,\frac{a_q}{p})$-partitioned by a $q$-fan.
\end{itemize}
\end{theorem}

\begin{proof}
The second part again follows from gnomonic projection, so we only show the first part.
We take the same configuration space and test maps as in the previous proof, only that we now have more possible solutions to exclude from $Z^m$.
In particular, let $L_i:=\{y\in\mathbb{R}^{p}: y_1+\ldots y_{a_1}=0, y_{a_1+1}+\ldots y_{a_1+a_2}=0,\ldots,y_{a_1+\ldots+a_{q-1}+1}+\ldots+y_p=0\}$.
Further let $\mathcal{L}_i:=\{L_i, \nu(L_i), \nu^2(L_i),\ldots,\nu^{p-1}(L_i)\}$.
We now want to show that there is no $\mathbb{Z}_p$-equivariant map $f:V_{2}(\mathbb{R}^d)\rightarrow Z^m\setminus\bigcup_{1\leq i\leq k}\mathcal{L}_i$.
By Lemma 6.1 in \cite{Barany2}, $f$ would induce a map $f': V_{2}(\mathbb{R}^d)\rightarrow M$, where $M$ is a $((p-1)k-2)$-dimensional manifold on which $\mathbb{Z}_p$ acts freely.
The non-existence of such a map now follows again from Lemmas \ref{lem:index}, \ref{lem:numerical_index} and \ref{lem:makeev}.
\end{proof}

The main result can now be proved using similar techniques.

\fans*

\begin{proof}[Proof of Theorem \ref{Thm:fans_main}]
The second part again follows from gnomonic projection, so we only show the first part.
Using the same configuration space and test map as above, we get a $\mathbb{Z}_k$-map $f: V_{2}(\mathbb{R}^d)\rightarrow S^{m(k-1)-1}$.
The issue is now that $\mathbb{Z}_k$ acts fixed-point-free, but not freely on $S^{m(k-1)-1}$: if, for example, for every $i$ and $j$ we have that $\mu_i(W_{p_1-j})=\mu_i(W_{2p_1-j})=\ldots=\mu_i(W_{k-j})$, the the subgroup of $\mathbb{Z}_k$ that is isomorphic to $\mathbb{Z}_{p_1}$ will leave the corresponding point on $S^{m(k-1)-1}$ invariant.
We call the subspace of all points of $S^{m(k-1)-1}$ which are of this form \emph{$p_1$-invariant}, and analogously for $p_2,\ldots,p_n$.
More generally, for any subset $P\subset\{p_1,\ldots,p_n\}=:\Pi$, we call a subspace of $S^{m(k-1)-1}$ \emph{$P$-invariant} if it is $p_i$-invariant for every $p_i\in P$.
Note that the action is fixed-point-free, so there is no $\Pi$-invariant subspace.
Finally, we call a subspace of $S^{m(k-1)-1}$ \emph{strictly $P$-invariant} if it is $P$-invariant but not $P'$-invariant for any $P'\supset P$.
We denote this subspace by $S_P$.

For each subset $P\subseteq\{p_1,\ldots,p_n\}$, we define a corresponding subgroup $\mathbb{Z}_k^P$ of $\mathbb{Z}_k$ as $\mathbb{Z}_k^P:=\bigtimes_{p_i\in P}\mathbb{Z}_{p_i}\cong \mathbb{Z}_a$, for $a=\prod_{p_i\in P}p_i$.
From the above definitions we have that a strictly $P$-invariant subspace $S_P$ is invariant under the action of the subgroup $\mathbb{Z}_k^P$.
On the other hand, the subgroup $\mathbb{Z}_k^{\bar{P}}$ (where $\bar{P}=\Pi\setminus P$) acts freely on $S_P$.
Note that as $p_1,\ldots,p_n$ are distinct primes we have that $\mathbb{Z}_k^P\times\mathbb{Z}_k^{\bar{P}}\cong\mathbb{Z}_k$.

We now have $S^{m(k-1)-1}=\bigcup_{P\subset\Pi}S_P$.
We will compute the cohomological index $\text{Ind}_{\mathbb{Z}_k}(S_P)$ for each $P\subset\Pi$.
By Lemma \ref{lem:index_product} and the above observation that $S_P$ is invariant under the action of $\mathbb{Z}_k^P$ we have $\text{Ind}_{\mathbb{Z}_k}(S_P)=(\text{Ind}_{\mathbb{Z}_k^{\bar{P}}}(S_P))\otimes H^*(B\mathbb{Z}_k^P)$.
It is known that for a finite cyclic group $\mathbb{Z}_l$ the cohomology ring of $B\mathbb{Z}_l$ is given by$H^*(B\mathbb{Z}_l)\cong \faktor{\mathbb{Z}[\alpha]}{(l\alpha)}$ (see e.g.\ Example 3.41 in \cite{Hatcher}, combined with the universal coefficients theorem).
In our case we further get
$$H^*(B\mathbb{Z}_k)\cong\bigotimes_{i=1}^n H^*(B\mathbb{Z}_{p_i})\cong\bigotimes_{i=1}^n \faktor{\mathbb{Z}[\alpha_i]}{(p_i\alpha_i)}\cong\faktor{\mathbb{Z}[\alpha_1,\ldots,\alpha_n]}{(p_1\alpha_1,\ldots,p_n\alpha_n)}\cong\faktor{\mathbb{Z}[\alpha]}{(k\alpha)},$$
where for the last isomorphism we can for example set $\alpha=\alpha_1+\ldots+\alpha_n$.
Similarly, the ring $H^*(B\mathbb{Z}_k^P)$ is generated by the $\alpha_i$'s that correspond to the $p_i$'s in $P$.

In order to compute $\text{Ind}_{\mathbb{Z}_k^{\bar{P}}}(S_P)$, we note that $\mathbb{Z}_k^{\bar{P}}$ acts freely on $S_P$ and thus, by the comment after Lemma \ref{lem:index}, we have $H^i(B\mathbb{Z}_k^{\bar{P}})\subseteq\text{Ind}_{\mathbb{Z}_k^{\bar{P}}}(S_P)$ for all $i>dim(S_P)\leq m(k-1)-1$.
We thus get that $\text{Ind}_{\mathbb{Z}_k}(S_P)$ contains, among others, all polynomials whose monomials all have degree larger than $m(k-1)-1$.
In particular, it contains polynomials that are not zero divisors in $H^*(B\mathbb{Z}_k)$, such as the polynomial $\alpha_1^{mk}+\ldots+\alpha_n^{mk}$.
It follows that $\prod_{P\subset\Pi}\text{Ind}_{\mathbb{Z}_k}(S_P)\neq 0$, and thus, by Lemma \ref{lem:index_union} $\text{Ind}_{\mathbb{Z}_k}(S^{m(k-1)-1})\neq 0$.
Lemma \ref{lem:index} now implies $i_{\mathbb{Z}_p}(S^{m(k-1)-1})\leq m(k-1)$.
On the other hand, from Corollary \ref{cor:makeev_general} we have $i_{\mathbb{Z}_k}(V_{2}(\mathbb{R}^d))=2d-2$.
By Lemma \ref{lem:index}, $f$ can only exist if $2d-2\leq m(p-1)$.
But $m$ was chosen such that $2d-3\geq m(p-1)$, so $f$ cannot exist.
\end{proof}

Instead of $k$-fans, we can apply the same idea to fans of $k$ double wedges.
More precisely, consider again the Stiefel manifold $V_{2}(\mathbb{R}^d)$ of all pairs $(x,y)$ of orthonormal vectors in $\mathbb{R}^d$.
To each $(x,y)\in V_2(\mathbb{R}^d)$ we assign a fan of $k$ double wedges $D(x,y)$ as follows:
We again consider the oriented linear subspace $h$ spanned by $(x,y)$.
Draw a line $\ell_1$ through $x$.
Then rotate another line $\ell_2$ in clockwise direction until the double wedge defined by $\ell_1$ and $\ell_2$ contains exactly a $\frac{1}{k}$-fraction of the first mass.
This can be continued, defining $k$ double wedges with a common apex, each of which contains exactly a $\frac{1}{k}$-fraction of the first mass.
As for $k$-fans, this construction continuously defines a fan of $k$ double wedges $D(x,y)$ through the origin for each $(x,y)\in V_2(\mathbb{R}^d)$.
Further, turning by two sectors defines a $\mathbb{Z}_k$-action on $V_{2}(\mathbb{R}^d)$ (turning $D(x,y)$ by one sector $k$ times, we retrieve $D(-x,-y)$, which is why we turn by two sectors; note that $k$ is odd).
The same arguments as above now give us the following:

\begin{corollary}
Let $k=p_1p_2\cdots p_n$ be a product of pairwise distinct odd primes.
\begin{itemize}
\item Let $2d-3\geq m(k-1)$.
Then any $m+1$ mass distributions $\mathbb{R}^d$ can be simultaneously $(\frac{1}{k}\ldots,\frac{1}{k})$-partitioned by a fan of $k$ double wedges through the origin.
\item Let $2d-1\geq m(k-1)$.
Then any $m+1$ mass distributions $\mathbb{R}^d$ can be simultaneously $(\frac{1}{k}\ldots,\frac{1}{k})$-partitioned by a fan of $k$ double wedges.
\end{itemize}
\end{corollary}

Given a fan of $k$ double wedges, we can perform a projective transformation which sends one of the hyperplanes to the hyperplane at infinity.
The apex, i.e., the intersection of all the hyperplanes, is then also sent to infinity.
In particular, the other hyperplanes are parallel after this projective transformation.
This proves the following:

\begin{corollary}
Let $k=p_1p_2\cdots p_n$ be a product of pairwise distinct odd primes.
Let $2d-1\geq m(k-1)$ and let $\mu_1,\ldots,\mu_{m+1}$ be $m+1$ mass distributions in $\mathbb{R}^d$.
Then there exists a projective transformation $\varphi$ such that $\varphi(\mu_1),\ldots,\varphi(\mu_{m+1})$ can be simultaneously $(\frac{1}{k}\ldots,\frac{1}{k})$-partitioned by $k-1$ parallel hyperplanes.
\end{corollary}

%
%
%

\bibliographystyle{plainurl} %
\bibliography{refs}

\begin{thebibliography}{10}

\bibitem{Barany4}
Imre B{\'a}r{\'a}ny, Pavle Blagojevi{\'c}, and Aleksandra~Dimitrijevi{\'c}
  Blagojevi{\'c}.
\newblock Functions, measures, and equipartitioning convex k-fans.
\newblock {\em Discrete \& Computational Geometry}, 49(2):382--401, 2013.

\bibitem{Barany3}
Imre B{\'a}r{\'a}ny, Pavle Blagojevi{\'c}, and Andr{\'a}s Sz{\H{u}}cs.
\newblock Equipartitioning by a convex 3-fan.
\newblock {\em Advances in Mathematics}, 223(2):579--593, 2010.

\bibitem{Barany2}
Imre B{\'a}r{\'a}ny and Ji\v{r}\'{i} Matou{\v{s}}ek.
\newblock Simultaneous partitions of measures by k-fans.
\newblock {\em Discrete \& Computational Geometry}, 25(3):317--334, 2001.

\bibitem{Barany}
Imre B{\'a}r{\'a}ny and Ji\v{r}\'{i} Matou{\v{s}}ek.
\newblock Equipartition of two measures by a 4-fan.
\newblock {\em Discrete {\&} Computational Geometry}, 27(3):293--301, 2002.

\bibitem{pizza_cccg}
Luis Barba and Patrick Schnider.
\newblock Sharing a pizza: bisecting masses with two cuts.
\newblock {\em CCCG 2017}, page 174, 2017.

\bibitem{Bereg2}
Sergey Bereg.
\newblock Equipartitions of measures by 2-fans.
\newblock In {\em International Symposium on Algorithms and Computation}, pages
  149--158. Springer, 2004.

\bibitem{Bereg}
Sergey Bereg, Ferran Hurtado, Mikio Kano, Matias Korman, Dolores Lara, Carlos
  Seara, Rodrigo~I. Silveira, Jorge Urrutia, and Kevin Verbeek.
\newblock Balanced partitions of 3-colored geometric sets in the plane.
\newblock {\em Discrete Applied Mathematics}, 181:21--32, 2015.

\bibitem{Blagojevic2}
Pavle~VM Blagojevic.
\newblock Topology of partition of measures by fans and the second obstruction.
\newblock {\em arXiv preprint math/0402400}, 2004.

\bibitem{pizza2}
Pavle~VM Blagojevi{\'c}, Aleksandra~Dimitrijevi{\'c} Blagojevi{\'c}, and Roman
  Karasev.
\newblock More bisections by hyperplane arrangements.
\newblock {\em arXiv preprint arXiv:1809.05364}, 2018.

\bibitem{Blagojevic}
Pavle~VM Blagojevi{\'c} and Aleksandra S~Dimitrijevi{\'c} Blagojevi{\'c}.
\newblock Using equivariant obstruction theory in combinatorial geometry.
\newblock {\em Topology and its Applications}, 154(14):2635--2655, 2007.

\bibitem{Karasev_logd}
Pavle~VM Blagojevi{\'c}, Roman Karasev, et~al.
\newblock Extensions of theorems of rattray and makeev.
\newblock {\em Topological Methods in Nonlinear Analysis}, 40(1):189--213,
  2012.

\bibitem{Borel}
Armand Borel.
\newblock Sur la cohomologie des espaces fibr{\'e}s principaux et des espaces
  homogenes de groupes de lie compacts.
\newblock {\em Annals of Mathematics}, pages 115--207, 1953.

\bibitem{Fadell}
Edward Fadell and Sufian Husseini.
\newblock An ideal-valued cohomological index theory with applications to
  borsuk—ulam and bourgin—yang theorems.
\newblock {\em Ergodic theory and dynamical systems}, 8(8*):73--85, 1988.

\bibitem{Hatcher}
Allen Hatcher.
\newblock {\em {Algebraic topology}}.
\newblock Cambridge Univ. Press, Cambridge, 2000.
\newblock URL: \url{https://cds.cern.ch/record/478079}.

\bibitem{pizza1}
Alfredo Hubard and Roman Karasev.
\newblock Bisecting measures with hyperplane arrangements.
\newblock {\em arXiv preprint arXiv:1803.02842}, 2018.

\bibitem{Jelic}
Marija Jeli{\'c}.
\newblock On {K}naster’s problem.
\newblock {\em Publications de l'Institut Mathematique}, 99(113), 2016.

\bibitem{Roman}
Roman {Karasev}.
\newblock personal communication, 2020.

\bibitem{KarasevSurvey}
Roman~Nikolaevich Karasev.
\newblock Topological methods in combinatorial geometry.
\newblock {\em Russian Mathematical Surveys}, 63(6):1031--1078, 2008.

\bibitem{Makeev}
Vladimir~V. Makeev.
\newblock {\em Universally inscribed and circumscribed polytopes}.
\newblock PhD thesis, St.Petersburg State University, 2003.

\bibitem{Matousek}
Ji\v{r}\'{i} Matou{\v{s}}ek.
\newblock {\em Using the Borsuk-Ulam Theorem: Lectures on Topological Methods
  in Combinatorics and Geometry}.
\newblock Springer Publishing Company, Incorporated, 2007.

\bibitem{spectral}
John McCleary.
\newblock {\em A user's guide to spectral sequences}.
\newblock Number~58. Cambridge University Press, 2001.

\bibitem{SoberonSurvey}
Edgardo Roldan-Pensado and Pablo Soberon.
\newblock A survey of mass partitions.
\newblock {\em arXiv preprint arXiv:2010.00478}, 2020.

\bibitem{HSSoCG}
Patrick Schnider.
\newblock Ham-sandwich cuts and center transversals in subspaces.
\newblock In {\em 35th International Symposium on Computational Geometry (SoCG
  2019)}. Schloss Dagstuhl-Leibniz-Zentrum fuer Informatik, 2019.

\bibitem{StoneTukey}
Arthur~H. Stone and John~W. Tukey.
\newblock Generalized ``sandwich'' theorems.
\newblock {\em Duke Math. J.}, 9(2):356--359, 06 1942.

\bibitem{tomDieck}
Tammo tom Dieck.
\newblock {\em Transformation groups}, volume~8.
\newblock Walter de Gruyter, 2011.

\bibitem{Handbook}
Csaba~D Toth, Joseph O'Rourke, and Jacob~E Goodman.
\newblock {\em Handbook of discrete and computational geometry}.
\newblock Chapman and Hall/CRC, 2017.

\bibitem{Volovikov}
A~Yu Volovikov.
\newblock On the index of g-spaces.
\newblock {\em Sbornik: Mathematics}, 191(9):1259, 2000.

\bibitem{ZivFans}
Rade~T Zivaljevic.
\newblock Combinatorics and topology of partitions of spherical measures by 2
  and 3 fans.
\newblock {\em arXiv preprint math/0203028}, 2002.

\end{thebibliography}

\end{document}